\documentclass[conference]{IEEEtran}
\IEEEoverridecommandlockouts
\usepackage{cite}
\usepackage{amsmath,amssymb,amsfonts}
\usepackage{algorithmic}
\usepackage{multirow}
\usepackage{graphicx}
\usepackage{textcomp}
\usepackage{xcolor}
\usepackage{color}
\usepackage{xspace}
\usepackage{url}
\def\BibTeX{{\rm B\kern-.05em{\sc i\kern-.025em b}\kern-.08em
   T\kern-.1667em\lower.7ex\hbox{E}\kern-.125emX}}
\usepackage{enumitem}
\usepackage{balance}
\usepackage{siunitx} %%% To line up numbers in tables

\usepackage{listings}
\usepackage{epsfig}
\usepackage{program}
\usepackage[vlined,linesnumbered]{algorithm2e}
\usepackage{soul}
\newtheorem{thm}{Theorem}

\begin{document}
\lstset{language=Java}

%% V1 %% \title{Fast Data Flow Testing Subsumption Discovery}
\title{{A Data Flow Analysis Framework for Data Flow Subsumption}} %V2
\author{\IEEEauthorblockN{Marcos Lordello Chaim}
\IEEEauthorblockA{University of Sao Paulo\\
Sao Paulo, SP, Brazil\\
Email:chaim@usp.br}

\and
\IEEEauthorblockN{Kesina Baral and Jeff Offutt}
\IEEEauthorblockA{George Mason University\\
Fairfax, VA, USA\\
Email:\{kbaral4,offutt\}@gmu.edu}
} % Switch for anonymous submission.
%\author{}  Anonymous submission.

\maketitle
\thispagestyle{plain}
\pagestyle{plain}

\begin{abstract}
Data flow testing creates test requirements as definition-use (DU) associations,
where a \textit{definition} is a program location that assigns a value to a variable
and a \textit{use} is a location where that value is accessed.
Data flow testing is expensive, largely because of the number of test requirements.
Luckily, many DU-associations are redundant in the sense that if one %DU-pair
test requirement (e.g., node, edge, DU-association) is covered,
other DU-associations are guaranteed to also be covered.
This relationship is called \textit{subsumption}.
Thus, testers can save resources by only
covering DU-associations that are not subsumed by other testing requirements.
%The subsumption relationship among testing requirements is a means to leverage structural testing. It allows  one to determine a subset of minimal testing requirements  whose property is to imply the coverage of all other testing requirements when they are covered. Testers could then focus on  the minimal testing requirements to create tests.  %Discovering the subsumption of data flow testing requirements, though, has shown to be an elusive goal. Previous solutions are costly and contain subtle flaws that might lead to incorrect results. 
In this work, we formally describe the \textit{Data Flow Subsumption Framework} (DSF) conceived to tackle the data flow subsumption problem. We show that DFS is a distributive data flow analysis framework which allows efficient iterative algorithms to find the Meet-Over-All-Paths (MOP) solution for DSF transfer functions. The MOP solution implies that the results at a point $p$ are valid for all paths that reach $p$.
We also present an algorithm, called Subsumption Algorithm (SA), that uses DSF transfer functions and iterative algorithms to find the \textit{local} DU-associations-node subsumption; that is, the set of DU-associations that are covered whenever a node $n$
is toured by a test. A proof of SA's correctness is presented and its complexity is analyzed. %By using SA one can efficient find the subsumption of DUAs by nodes (DUA-node subsumption), by edges (DUA-edge subsumption), and by DUAs (DUA-DUA subsumption).

\end{abstract}

%%% Data flow testing creates test requirements as definition-use (DU) associations, where a \textit{definition} is a program location that assigns a value to a variable and a \textit{use} is a location where that value is accessed. Data flow testing is expensive, largely because of the number of test requirements. Luckily, many DU-associations are redundant in the sense that if one test requirement (e.g., node, edge, DU-association) is covered, other DU-associations are guaranteed to also be covered. This relationship is called \textit{subsumption}. Thus, testers can save resources by only covering DU-associations that are not subsumed by other testing requirements.

%By finding subsumption relationships,
%testers can find a minimal subset of DU-pairs such that if a test set covers all minimal DU-pairs,
%the test set is guaranteed to also cover all DU-pairs.

%Although this has the potential to significantly decrease the cost of data flow testing, finding subsumption among DU-associations is quite difficult. Previous solutions are costly and contain subtle flaws that sometimes lead to incorrect results. We model the data flow testing subsumption as a data flow analysis framework, allowing us to use efficient algorithms that quickly discover data flow subsumption relationships. Experimental data suggest that the framework and algorithm can reduce the cost of data flow testing and will work at scale. 


%For change tracking: \\
%\marcos{Marcos} ~~~~~~~~~~
%\jeff{Jeff} ~~~~~~~~~~
%\kesina{Kesina}

%%%%%%%%%%%%%%%%%%%%%%%%%%%%%%%%%%
%%%%%%%%%%%%%%%%%%%%%%%%%%%%%%%%%%
\section{Introduction}
\label{sec:intro}
Data flow testing (DFT) attempts to enable
comprehensive structural testing based on flows of data through software
\cite{las83:ing,nta84:ing,rap85:ing,ura88:ing}.
Roughly speaking, it involves developing tests that
exercise (cover) every value assigned to a
variable and its subsequent references (uses).
These pairs of definitions and uses are called
\textit{definition-use associations} (DUA)
and the paths from defs to uses are called
\textit{DU-paths}
\cite{rap85:ing}.
DFT {uses} both control- and data flow information to design tests.
As a result, data flow tests can exercise more situations than control-flow testing can.
The intuition is that causing more values to reach different uses
can find more problems in the software and increase confidence in its reliability.

Studies have shown that DFT can effectively detect faults in programs \cite{hut94:ing,fra98:ing} and can verify the security of web applications \cite{dao11:ing}.
Hemmati \cite{hem15:ing} conducted a study in which control-flow criteria, namely, statement, branch, loop, and MCDC coverage, were compared against definition-use pair coverage with respect to their
ability to detect faults.
They found that out of 274 faults in sizable open-source programs,
only 76 (28\%) were found by control-flow coverage criteria.
For those same faults,
definition-use pair coverage detected 79\% of the faults not detected by control-flow criteria.
%%Can we say which control-flow criteria?
%%Node and edge coverage are different from edge-pair, and of course, prime path coverage subsumes ADUP
Thus, DFT can help
achieve and
verify software quality,
which is especially important for mission-critical systems.

To achieve high DFT coverage, a tester needs to develop specific test cases to cover a large number of DUAs.
Furthermore, some DUAs cannot be covered by any test cases because the underlying path is infeasible.
%%%Furthermore, the DUAs unverified by test cases should be examined to determine whether they are infeasible\footnote{An infeasible testing requirement is such that there are not input data traversing paths that exercises it.}.
Both tasks require human intervention,
which increases the cost of DFT\@.

Many approaches to reduce DFT's cost have been created.
Some exploit the subsumption relationship among DUAs \cite{mar03:ing,ji18}.
A \textit{test requirement} (TR) $tr_1$ (e.g., a DUA $D_1$) \textit{subsumes} another test requirement $tr_2$ (another DUA $D_2$) if every complete path that traverses $tr_1$ also traverses $tr_2$.
The minimal subset of TRs that subsumes every other TR is called a \textit{spanning set} and
its elements are referred to as \textit{unconstrained} test requirements \cite{mar03:ing}.

If a spanning set of DUAs could be identified,
testers would only need to satisfy the unconstrained DUAs,
saving time and effort.
Jiang et al. \cite{ji18} showed that
targeting unconstrained DUAs
reduces the cost of input data generation.
Additionally,  DUAs can be subsumed by 
test requirement (e.g., node, edge) of different criteria so that once they are covered some DUAs are guaranteed to also be covered \cite{san07:ing}.

%Discovering  data flow testing subsumptions (DFTS), though, has been an elusive goal. 
%Current algorithms are costly, being linear \cite{san07:ing} and quadratic \cite{mar03:ing,ji18} to the number of DUAs for subsumption of DUAs by edges (DUA-edge subsumption) and by DUAs (DUA-DUA subsumption), respectively, which has hampered their application for industry-size systems. 
%Furthermore, some of them \cite{mar03:ing,ji18} might miss paths that would block the subsumption of DUAs, leading  sometimes to incorrect results.

We present a novel and efficient approach to tackle data flow subsumptions. 
It models the problem of finding the \textit{local DUA-node} subsumption; that is, those DUAs that are covered whenever a particular point $p$ (e.g, a node) of a program is reached, as a data flow analysis framework \cite{hec77:ing,kam77:ing}.
Using the local DUA-node subsumption, one can efficiently discover the subsumption of DUAs with respect to nodes, edges, 
%pair of edges, 
and other DUAs \cite{cha20a}. 

The goal of this work is to describe formally the Data Flow Subsumption Framework (DSF), showing that DFS is  distributive and that it can be used to  solve the local DUA-node subsumption. 
A distributive framework allows iterative algorithms to find the Meet-Over-All-Paths (MOP) solution, which means that the results are valid for all paths that reach a program point $p$. We also present an algorithm, called Subsumption Algorithm (SA), that uses DSF transfer functions and iterative algorithms to find the \textit{local} DUA-node subsumption. A proof of SA's correctness is presented and its complexity is analyzed.

This document starts with background in data flow testing in section \ref{sec:back}.
We then describe  the local DUA-node subsumption problem in section~\ref{sec:localduanode}. Section~\ref{sec:dsf} describes the Data Flow Subsumption Framework (DSF),
followed by the algorithm to solve  the local DUA-node subsumption in section \ref{sec:sa}.
%Related work is discussed in section \ref{sec:related}, and 
We draw the conclusions  in section \ref{sec:concl}.

%%%%%%%%%%%%%%%%%%%%%%%%%%%%%%%%%%
%%%%%%%%%%%%%%%%%%%%%%%%%%%%%%%%%%
\section{Background}
\label{sec:back}

{Graph based testing
{criteria use graph abstractions}
of the software under test to generate tests.
A graph can be defined as $G(N, E, s, e)$, where $N$ is a set of nodes, $E$ is a set of edges, $s$ is the start node and $e$ is the exit node.
A node $(n)$ can represent a single statement of the program or a sequence of statements.
For our purposes, we consider a sequence of statements, also known as a \textit{basic block}, as a node.
{An edge represents potential control flow from one node to another,
written as ($n_{i}$, $n_{j}$),  {$n_i \neq n_j$},
where node $n_{i}$ is the \textit{predecessor}
and node $n_{j}$ is the \textit{successor}.}
%% If two nodes have a potential for control transfer, they form an edge.
%% An edge is written as ($n_{i}$, $n_{j}$), where the control transfers from $n_{i}$ to $n_{j}$.
Graphs extracted from a program must have at least one start node and exit node for it to be useful to generate tests.
A program can have multiple entry and exit points.

A \textit{path} is a sequence of nodes
($n_i$, $\ldots$, $n_k$, $n_{k+1}$, $\ldots$, $n_j$),
where $i \leq k < j$,
such that ($n_k, n_{k+1}$) $\in E$\@.
%%% Moved to be with def of edge
%%% A node $n_k$ is said to be a \textit{predecessor} of a node $n_{k+1}$ and $n_{k+1}$ is said to be a \textit{successor} of $n_k$, if there exists an edge between the two nodes.
A \textit{test path} is a
{special}
path that starts from a start node $s$ and ends at an exit node $e$.
A test path represents the execution of
{one or more}
test cases.
{A \textit{side-trip} is a sub-path that starts and ends at the same node
(a loop).}

Figure~\ref{fig:progmb} presents a program
that finds the maximum element in an array of integers \cite{cha13:ing} and Figure~\ref{fig:graphmba} presents its control flow graph.
The numbers at the start of each line of code in Figure~\ref{fig:progmb} indicate the line's corresponding node in the graph.

%% (Marcos) Paragraph below removed to save space. Since node 6 is associated with the last bracket of the program, there is an explanation regarding its presence in the flow graph.

%\jeff{Our data flow algorithms assume that each graph has a single entry and a single exit node.
%Thus an additional node is added to represent a single exit.
%If the software has multiple exit points
%(for example, two return statements),
%our tool adds edges from all exit locations to the special exit node.
%Figure~\ref{fig:graphmba} shows the control flow graph for Max,
%annotated with defs and uses on the nodes.
%Node 6 is the added exit node.}

{Graph coverage criteria come in two forms,}
control flow coverage criteria and data flow coverage criteria.
\textit{Control flow coverage criteria} %, also known as structural coverage criteria,
cover the structure of the graph,
{including}
nodes, edges,
{and specific sub-paths.}
\textit{Data flow coverage criteria} evaluates the flow of data values during program execution.
Data flow coverage criteria provide test requirements for data flow testing by focusing on definitions and uses of variables.
A \textit{definition}, or \textit{def},
is a program location where a value is assigned to a variable.
A \textit{use} is a location where the variable is referenced.
The graph shown in Figure~\ref{fig:graphmba} is annotated with defs and uses associated with its nodes and edges.
} %% end kesina{}

\begin{figure}[!htbp]
\centering
\small
\begin{lstlisting}
/* 0 */ int max(int array [], int length)
/* 0 */ {
/* 0 */    int i = 0;
/* 0 */    int max = array[++i];
/* 1 */    while (i < length)
/* 1 */    {
/* 3 */       int rogue = 1
/* 3 */       if (array[i] > max)
/* 5 */          {
/* 5 */           max = array[i];
/* 5 */           print(rogue);
/* 5 */          }
/* 4 */       i = i + 1;
/* 4 */    }
/* 2 */    return max;
/* 6 */ }
\end{lstlisting}
\vspace{-12pt} %% reduce unnececessary gap
\caption{Example program Max.\label{fig:progmb}}
\end{figure}

\begin{figure}[!htbp]
  \centering
  \includegraphics[scale=0.30]{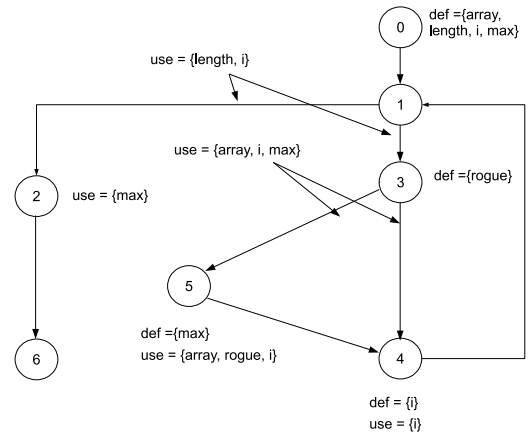}
  \caption{Annotated flow graph for max.eps\label{fig:graphmba}}
\end{figure}

% Data flow testing requires that test cases exercise paths in
% a program between locations where a value is assigned to a variable
% (called \textit{definitions}, or \textit{defs})
% and
% its subsequent references
% (called \textit{uses}).
% %% When a variable receives a new value,
% %% it is called a \textit{definition}; when a variable is referenced, it is called a \textit{use}.
% A variable can be used to compute a value or to compute a predicate. Value computations are associated with nodes and predicate computations are associated with edges.
% %\st{When referred to in a predicate computation, it is called a \textit{p-use} and is associated with edges; otherwise, it is called a \textit{c-use} and is associated with nodes.}

% Figure~\ref{fig:graphmba} gives the flow graph of the Max program annotated with definitions (def) and uses associated with nodes and edges.
% A \textit{definition-clear} (def-clear) path with respect to a variable $x$ is a path where $x$ is not redefined in any node in the path,
% \kesina{except possibly in the first and last nodes??}. \kesina{ A def-clear path allows side-trips.}
% \kesina{A \textit{du-path} is a simple subpath that is def-clear with respect to variable $x$.}
% % A \textit{du-path} with respect to variable $X$ is a def-clear path that is also a simple path.
% % A def-clear path with respect to $x$ is a du-path that allows def-clear side-trips;
% % that is, the nodes between the first and last nodes may occur several times as long as they do not redefine $X$.

%%%-------------------------------%%%
%\subsection{\kesina{Data flow testing}}
{Data flow testing focuses on the flow of data values from definitions to uses.}
%%%% Defined in the previous paragraph ...
%%% between locations where a value is assigned to a variable
%%% (called \textit{definitions}, or \textit{defs})
%%% and
%%% its subsequent references
%%% (called \textit{uses}).
% \kesina{Pairs of defs and uses are called \textit{du associations.} }
%% When a variable receives a new value,
%% it is called a \textit{definition}; when a variable is referenced, it is called a \textit{use}.
{A variable can be used to compute a value or in a predicate.
Value computations are associated with nodes and predicate computations are associated with edges.}

A \textit{definition-clear} (\textit{def-clear})
path with respect to a variable $x$ is a path where $x$ is not redefined
{along the path.}
% in any node in the path,
{A \textit{du-path} is a simple sub-path (all nodes are different except the first and last nodes) that is def-clear with respect to (wrt) variable $x$.}  {A  du-path with side-trips wrt variable $x$ allows side-trips that are also def-clear wrt $x$.} 
% A \textit{du-path} with respect to variable $X$ is a def-clear path that is also a simple path.
% A def-clear path with respect to $x$ is a du-path that allows def-clear side-trips;
% that is, the nodes between the first and last nodes may occur several times as long as they do not redefine $X$.

{Data flow test criteria define test requirements as specific du-paths that must be covered.}
{A \textit{DU-association} set $D(d,u,x)$ is a set of du-paths {and du-paths with side-trips} {wrt} variable $x$ that start at node $d$ and end at node $u$.
{If the use is on an edge ($u', u)$,
the DU-associations set is written as $D(d, (u', u), x)$.}
Several data flow testing criteria
{have been invented}
\cite{las83:ing,nta84:ing,rap85:ing,ura88:ing}.
{In this paper,}
we focus on the \textit{all-uses} criterion proposed by Rapps and Weyuker \cite{rap85:ing}.}

% Many data flow testing criteria have been proposed in the literature \cite{las83:ing,nta84:ing,rap85:ing,ura88:ing}. Nevertheless, we will focus on the all-uses criterion proposed by Rapps and Weyuker \cite{rap85:ing}. All-uses  requires that \textit{definition-use associations} (DUAs) be covered.
% The triple $D$ = ($d$, $u$, $X$)
% %% \st{, called c-use DUA,}
% represents a data flow testing requirement involving a definition in
% node $d$ and a
% %% \st{c-use}
% use in node $u$ of variable $X$ such that there is a
% def-clear path wrt $X$ from $d$ to $u$.
% Likewise, the triple $D$ = ($d$, ($u'$, $u$), $X$)
% %% \st{, called p-use DUA,}
% represents the association
% between a definition and a
% %% \st{p-use}
% use of a variable $X$ in edge ($u'$, $u$).
% In this case,
% a def-clear path ($d$,$\ldots$,$u'$,$u$) wrt $X$ should exist.

\begin{table}[ht]
\begin{center}
\caption{All-uses test requirements for program Max.} %\vspace{0.2cm}
%\vspace*{0.2cm}
\label{tab:DUAs}
\resizebox{\linewidth}{!}{
\begin{tabular}{|c|c|c|} \hline
\multicolumn{3}{|c|}{All-uses} \\ \hline
(0, (3,5), array) &
(0, (3,4), array) &
(0, 5, array)  \\ \hline
(0, (1,3), length) &
(0, (1,2), length) &
(0,	(1,3), i) \\ \hline
(0, (1,2), i)&
(0, (3,5), i) &
(0, (3,4), i) \\ \hline
(0, 4, i) &
(0, 5, i) &
(0, 2, max) \\ \hline
(0, (3,5), max)  &
(0, (3,4), max) &
(5, 2, max) \\ \hline
(5, (3,5), max) &
(5, (3,4), max) &
(4, (1,3), i) \\ \hline
(4, (1,2), i) &
(4, (3,5), i) &
(4, (3,4), i) \\ \hline
(4, 4, i) &
(4, 5, i) &
(3, 5, rogue)
\\ \hline
\end{tabular}}
\end{center}
\end{table}

{The all-uses criterion requires that at least one du-path {(or du-path with side-trips)} is
{executed, or \textit{toured},}
for every DU-associations (DUAs) set,
that is,
{each def reaches each use at least once.}
If a test set $T$ includes a du-path {(or du-path with side-trips)} for each DUA
$D$($d$, $u$, $x$) or $D$($d$, ($u'$, $u$), $x$),
it is said to be \textit{adequate}
for the all-uses criterion for program $P$ since all required DUAs were \textit{covered}.}

\section{DUA-node subsumption}
\label{sec:localduanode}

DUA-node subsumption {identifies DUAs that are guaranteed to be covered if a specific node in the graph is visited.
More formally,
DUA $D(d, u, X)$ is \textit{DUA-subsumed} by node $n$
if $D$ is covered on {a} test path that visit node $n$  
{and reach the exit node}.
The set of DUAs subsumed by node $n$
is the set of all DUAs that are covered by all test paths that visit $n$.}

{We find it necessary to allow for interrupted execution,
for example exceptions or other program aborts.
Thus, we distinguish between
\textit{local DUA-node subsumption,}
which is the set of DUAs covered
by all paths that \textbf{reach} $n$,
and
\textit{global DUA-node subsumption,}
which is the
set of DUAs covered by all test paths that both reach $n$
and then continue to the exit node.
The set of globally subsumed DUAs
include DUAs that are DUA-node subsumed by nodes that appear on all paths from node $n$ to the exit node.}

{Figure~\ref{fig:DUAnodeMax}
shows the DUA-subsumption sets for the Max method from
Figures \ref{fig:progmb} and \ref{fig:graphmba}.
Each node contains the locally subsumed DUAs.
For example, if node 5 is reached,
the definition of $array$ at node 0
is guaranteed to have reached the use on edge (3,5).}

\begin{figure}
    \centering
    \includegraphics[scale=0.3]{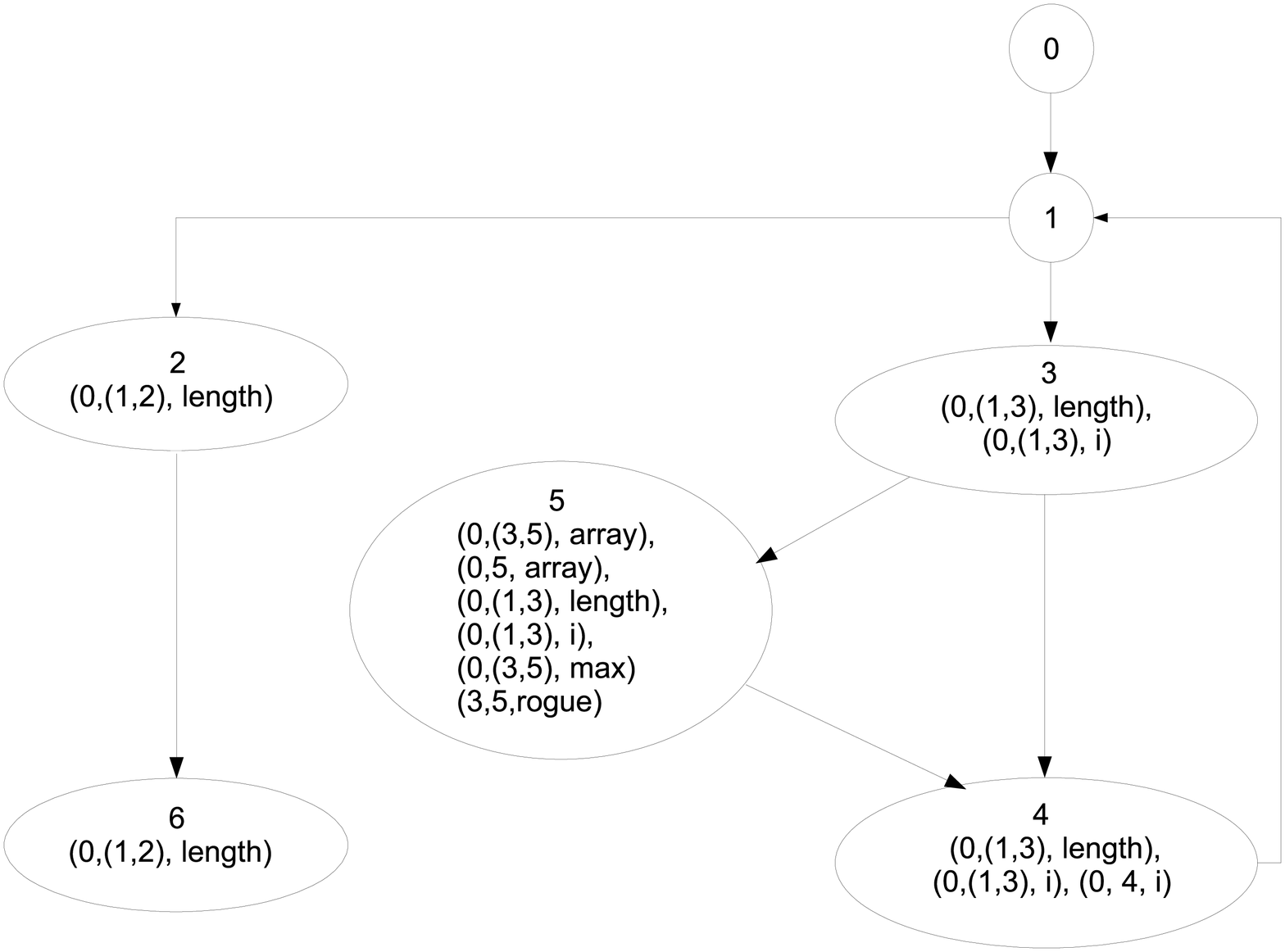}
    \caption{Local DUA-node subsumption for program Max}
    \label{fig:DUAnodeMax}
\end{figure}

{Node $n_i$ \textit{dominates} node $n_j$ if every path from the start node to $n_j$ includes $n_i$ \cite{hec77:ing}.
{
Node $n_j$ \textit{post-dominates} node $n_i$ if any path from $n_j$ to the exit node includes $n_i$.}
A node dominates itself but does not post-dominate itself \cite{ji18}.
In the absence of early program termination,
node 5 post-dominates nodes 4, 1, 2, and 6.
When they are visited by a test path,
the set of DUAs that are globally subsumed by node 5
includes the six DUAs listed in node 5,
plus DUAs (0, 4, i) from node 4 and (0, (1,2), length) from node 2.
Thus, node 5 locally subsumes six DUAs
and globally subsumes eight DUAs.}

Node 5 is the only unconstrained node for program Max.
%; that is, if it is covered then all nodes will eventually be covered, excepting if an abort command is executed after touring node 5.
This means that eight of 24 DUAs will be covered if all nodes of the Max program are visited.
Thus, node coverage would result in a data flow coverage of 33\%.

In the next section, we present a data flow analysis framework that allows one to find the local DUA-node subsumption.

%%%-------------------------------%%%

%%%%%%%%%%%%%%%%%%%%%%%%%%%%%%%%%%
%%%%%%%%%%%%%%%%%%%%%%%%%%%%%%%%%%
\section{Data flow subsumption framework}
\label{sec:dsf}
Data flow analysis frameworks were devised to solve data flow analysis problems such as reaching definitions, live-variables, and available expressions \cite{kam77:ing}.  Their goal is to determine a \textit{fact} that is valid at the entry or exit of a program point $p$ whenever $p$ is reached \cite{hor87:ing}.  We describe below the \textit{Data Flow Subsumption Framework} (DSF), which can be used to find data flow testing subsumption relationships.

The fact that  DSF  determines is the set of  DUAs  already \textit{covered}  or \textit{available} to be covered at the entrance of  a node $n$ and at the exit of $n$.
In doing so, the DUAs covered along all paths that reach a node $n$---and also after its  traversal---are discovered, as well as  the DUAs  available for coverage in these very same paths. A DUA is available  for coverage in a path if its def node was toured previously in the path and there is no redefinition of its associated variable in the subsequently toured nodes of the path.

In what follows, we present the  formal description of DSF. We start off by briefly discussing the definitions regarding data flow analysis frameworks. Next we show DSF transfer functions and its properties. We then present the Subsumption Algorithm (SA) which solves DFS and finds all DUAs covered by any path that reach a particular node $n$; that is, the local DUA-node subsumption. SA is an adaptation of interactive algorithms to solve data flow analysis frameworks \cite{mar90:ing}. We finish up the section by analyzing SA's complexity.

%%%-------------------------------%%%
\subsection{Definitions}

The goal of data flow analysis is to solve data flow problems by assigning  a program fact to each node of the flow graph that will be valid every time the node is reached during every possible execution \cite{hor87:ing}.
Data flow analysis problems can be \textit{forward} or \textit{backward} problems. A forward problem gives the fact valid at the entrance of a particular node whilst backward problems at the exit of the node.  Available expressions and reaching definitions are examples of forward problems; live-variable analysis is a backward problem.  Data flow analysis frameworks model  data flow problems so that particular algorithms are able to find the fact assigned to each node.

The domain of program facts in  this work
is modeled as a bounded (contains no infinite chains)  semi-lattice $V$ with meet\footnote{Meet is a binary operation for which the idempotent, commutative, and associative properties hold. } operation $\wedge$, least element $\perp$ (bottom), greatest element $\top$ (top),  and a partial order $\leq$. The top element is  such that $x \wedge \top$ = $x$ for any element $x$ in $V$.  The least element $\perp$ (bottom) is such that for any x element of $V$,  $x \wedge  \perp $ = $\perp$.

The relationship between values of $V$ before and after a node $n$ of the flow graph is given by the transfer functions of the framework. The family of transfer functions $F: V \rightarrow V$ in a data flow analysis framework has the following properties: (1) $F$ has an identity function $I$, such that $I(x) = x$ for all $x$ in $V$; (2) $F$ is closed under composition; that is, for any two functions $f$ and $g$ in $F$, the function $h$ defined by $h(x) = g(f(x))$ or $h(x) = g \circ f(x)$ is in $F$ \cite{lam07:ing}. A function $f: V \rightarrow V$ is \textit{monotone} iff for all $x$ and $y$ in $V$, $x \leq y$ implies $f(x) \leq (y)$.

In this work, a \textit{monotone data flow analysis framework}  is composed of the following components:
\begin{enumerate}
\item A data flow graph, with specially labeled ENTRY and EXIT node;
\item  a direction (forward or backward) of the data flow;
\item a set of values $V$;
\item a meet operator $\wedge$ and partial order $\leq$;
\item a set of functions $F$ such that the identity function belongs to $F$ and all $f_n$ in $F$, where $f_n$ is the transfer function of node $n$, are monotone and closed under composition; and
\item a constant value $v_{ENTRY} \in \{\perp,\top\}$  or $v_{EXIT} \in \{\perp,\top\}$, representing the boundary condition for a forward or backward framework, respectively.

\end{enumerate}

We denote a function $f_\pi$ associated with path $\pi$, as follows: if $\pi$ is the empty path then $f_\pi(x) = x$ (identity function); if $\pi = (n_1, n_2, \ldots,n_{k-1}, n_k$) then $f_\pi(x) = f_{n_k} f_{n_{k-1}}\ldots f_{n_2} f_{n_1}(x)$.

In general, algorithms that solve data flow analysis problems with data flow analysis frameworks find the \textit{maximum fixed point} (MPF) solution. Let MPF($n$) be the \textit{maximum fixed point} solution for node $n$, then MPF$(n)$ $\leq  f_p(\text{MPF}(p))$ where $p$ is a predecessor of $n$. The ideal result for a data flow analysis framework, though, is the \textit{meet-over-all-paths} solution, the map MOP($n$) = $\wedge$ \{$f_\pi(v_{ENTRY}) \mid \pi$ is a path from $s$ to $n$\} \cite{hor87:ing}.

The MOP solution means that  the meet operation is applied to the composition of the transfer functions along all paths that can reach $n$. Unfortunately, the MOP solution is in general undecidable because a flow graph with cycles may have an unbounded number of paths. Luckily, it can be obtained if the transfer functions $f_n$ of $F$ of the data flow analysis framework are distributive; that is, $f_n(x \wedge y) = f_n(x) \wedge f_n(y)$\footnote{A distributive function is also monotone \cite{lam07:ing}.}. A monotone data flow analysis framework whose transfer functions is distributive is called \textit{distributive data flow analysis framework}.

In the case of DSF, the data flow graph is the flow graph $G(N,E,s,e)$ where ENTRY is the start node $s$ and EXIT is the exit node $e$.  The direction of the data flow is forward and the values of $V$ comprises all subsets obtained from all DUAs required to test a program $P$. The meet operation is set intersection and the partial order is defined by the  $\subseteq$ operation. The $v_{ENTRY}$ is $\varnothing$ because at the entrance of the start node $s$ there is no DUA covered or available for coverage.

We present the DSF transfer functions as follows.

%%%-------------------------------%%%
\subsection{Transfer functions}

The transfer functions are a pivotal point of a data flow analysis framework since they ultimately  calculate the fact that is valid at the entrance and exit of a node $n$ when it is reached. To define the transfer functions of  DSF, we associate sets with each node of the flow graph. They were originally  introduced by Chaim and Araujo \cite{cha13:ing} and are defined as follows:

 Let $n \in N$ be a node of a flow graph $G(N,E, s, e)$ of a program $P$ and ($d$,$u$, $X$) or ($d$,($u'$,$u$), $X$)  a DUA required for testing $P$ according to the all-uses criterion.

\begin{description}

\item [Born($n$)]: set of DUAs ($d$,$u$, $X$) or ($d$,($u'$,$u$), $X$)   such that  $d = n$. It encompasses DUAs that become available after the traversal of  its def node.

\item[\textbf{Disabled}($n$)]: set of DUAs ($d$,$u$, $X$) or ($d$,($u'$,$u$), $X$) such that $X$ is defined in $n$ and  $d\neq n$. Set of DUAs whose variable $X$ has been defined at node $n$ so that these DUAs become disabled (or unavailable) for coverage after the traversal of $n$.

\item [PotCovered($n$)]:  set of DUAs  ($d$,$u$, $X$) or ($d$,($u'$,$u$), $X$) so that $u = n$. It comprises DUAs that are potentially covered at $n$ provided they are available; that is, their def node has been previously traversed.

\item [Sleepy($n$)]: set of DUAs ($d$,($u'$,$u$), $X$)  such that  $u' \neq n$. DUAs that are temporarily unavailable (sleepy) because they cannot be covered immediately after the of traversal $n$. It includes DUAs such that the use occurs in an edge ($u'$, $u$) where $u' \neq n$.

\end{description}

\textbf{Born}($n$) sets are similar to the \textbf{gen}($n$) and \textbf{e\_gen}($n$) sets of the reaching definitions and available expressions problems \cite{lam07:ing}. They represent those DUAs that are \textit{born} because the node where their variable is assigned has been toured. Likewise, \textbf{Disabled}($n$) is analogous to \textbf{kill}($n$) and \textbf{e\_kill}($n$) sets of the same data flow analysis problems since they contain those DUAs that are \textit{killed} after the traversal of $n$.

Differently of other data flow analysis problems, DSF needs sets \textbf{PotCovered}($n$) and \textbf{Sleepy}($n$). The first sets represent those DUAs that can \textit{potentially} be covered when a node is traversed. If a DUA is avaliable for coverage when $n$ is reached and it belongs to \textbf{PotCovered}($n$), then it will be covered after the traversal of $n$.

\textbf{Sleepy}($n$) sets aim at blocking some edge DUAs ($d$,($u'$,$u$
),$X$) of being covered after the traversal of a node $n$. For instance, \textbf{Sleepy}(3) comprises all edge DUAs excepting those whose use is either in edge  (3,5) or (3,4), which means that only DUAs with these edges will be allowed to be covered after the traversal of node 3.
Consider that a test has toured the  path (0, 1, 3, 5) of program Max. The next node to be toured is node 4.  \textbf{Sleepy}(5),  according to the definition, contains all edge DUAs required by  Max because there is no DUA with an use in edge (5,4) to be spared. So  when node 4 is toured, \textbf{Sleepy}(5)---the \textbf{Sleepy} set of its predecessor---can be used to block the coverage of edge  DUAs such as (0, (3,4), array) or (4, (3,4), i) by path (0, 1, 3, 5, 4), which can  potentially be covered at node 4, but when the predecessor is node 5 they cannot.

Additionally, the transfer functions of DSF utilizes  two working sets---the current sleepy DUAs (\textbf{CurSleepy}) and the  covered DUAs at $n$  (\textbf{Covered}($n$))---in their definition. Below we present the transfer functions (referred to as TF1 and TF2) and two auxiliary functions (referred to as AF1 and AF2) used to calculate \textbf{CurSleepy} and  \textbf{Covered}($n$).

\vspace*{.25cm}

\begin{enumerate}
\item[\textit{TF1}] $\text{\textbf{IN}}(n) = \bigcap_{p \in PRED(n)} \text{\textbf{OUT}}(p)
$

\item[\textit{AF1}] $\text{\textbf{CurSleepy}} = \bigcup_{p \in PRED(n) \text{and} (p,n) \text{is not a back edge}} \text{\textbf{Sleepy}}(p)
$

\item[\textit{AF2}] $ \text{\textbf{Covered}}(n) =
\bigcap_{p \in PRED(n)} \text{\textbf{Covered}}(p) \cup [(\text{\textbf{IN}}(n) - \text{\textbf{CurSleepy}}) \cap \text{\textbf{PotCovered}}(n)]$

\item[\textit{TF2}] $\text{\textbf{OUT}}(n) = \text{\textbf{Born}}(n) \cup [\text{\textbf{IN}}(x) - \text{\textbf{Disabled}}(n)] \cup \text{\textbf{Covered}}(n)$
\end{enumerate}
where PRED($n$) is the set of nodes of $G$ that are predecessors of node $n$.

\vspace*{.25cm}

%Explain the transfer functions.
The interactive algorithms that solve data flow analysis problems  find the values (facts) associated with sets  \textbf{IN}($n$) (forward problems) and with sets \textbf{OUT}($n$) (backward problems). In DSF, \textbf{IN}($n$)  contains the set of  DUAs  already \textit{covered}  or \textit{available} to be covered at the entrance of a node $n$ when it is reached by any path. \textbf{OUT}($n$), in turn, contains those DUAs covered or available for coverage  after the traversal of $n$ by these paths.
%We indeed are interested in the value of the working set \textbf{Covered}($n$)---the DUAs covered after touring $n$---for the purposes of data flow subsumption.
The meet  operation of DSF is the set intersection and it is used to find the value of \textbf{IN}($n$)  by intersecting  \textbf{OUT} sets of the predecessors of $n$ in TF1.

Auxiliary function AF1 calculates the edge DUAs that cannot be covered at node $n$. \textbf{CurSleepy} is the union of the DUAs blocked after the tour of a predecessor $p$ of $n$ provided ($p$,$n$) is not a back edge. A back edge ($p$,$n$) is such that node $n$ dominates node $p$ \cite{lam07:ing}. In Figure
~\ref{fig:graphmba}, edge (4,1) is a back edge. \textbf{CurSleepy} is used to block   edge DUAs of being covered when one cannot predict from which path a node $n$ is reached. However, when ($p$,$n$) is a back edge, one knows that  $n$ is always toured before touring $p$ so that it does not block other edge DUAs of being covered at  $n$.

Consider node 4 of the Max program; it can be reached by nodes 3 and 5 through non-back edges (3,4) and (5,4). However, one does not know from which predecessor it has been reached; thus, no edge DUA with edges (3,4) and (3,5) can be covered at node 4 by all paths that reach 4. \textbf{CurSleepy} at node 4 is the union of \textbf{Sleepy}(3) and \textbf{Sleepy}(5). \textbf{Sleepy}(3) comprises all edge DUAs excepting those whose uses occur in edges (3,5) or (3,4); and \textbf{Sleepy}(5) comprises all edge DUAs as mentioned before. As a result, \textbf{CurSleepy} will contain all edge DUAs; thus, none would be covered at node 4 according to AF1.

Auxiliary function AF2 finds the DUAs that are covered by all paths that reach node $n$; it is divided in two parts that are added by a union set operation. The first part of AF2  intersects \textbf{Covered} sets of the predecessors of $n$. In doing so, node $n$ will \textit{inherit} only DUAs that are covered previously in all paths that reach it. The second part of AF2 calculates the DUAs covered at node $n$. \textbf{IN}($n$) has the DUAs covered and available to be covered in all paths that reach $n$ according to TF1; \textbf{CurSleepy} has the edge DUAs that are blocked at node $n$; and \textbf{PotCovered}($n$) contains DUAs that might be covered at $n$ provided they are in \textbf{IN}($n$). The operations described in the second part of AF2 determine the DUAs covered at $n$: \textbf{CurSleepy} DUAs are removed from \textbf{IN}($n$) and the result is intersected with \textbf{PotCovered}($n$). The remaining DUAs after these operations along with DUAs covered in previously toured nodes give the DUAs covered at node $n$.

Finally, transfer function TF2 determines the \textbf{OUT}($n$) sets; that is, those DUAs that are \textit{forwarded} in the data flow analysis. They are calculated in three parts that are added by union set operations. The first part is the \textbf{Born}($n$) set, which contains the DUAs that become available for coverage at node $n$; that is, their variable is assigned at $n$. The second part is composed of those DUAs that are available in \textbf{IN}($n$) and \textit{survive} node $n$  because  they do not belong to \textbf{Disabled}($n$). The last set added in TF2 is the set of DUAs covered at $n$ (\textbf{Covered}($n$)). All these DUAs are forwarded to its successors in the data flow analysis.

The transfer functions of DSF are only useful if they satisfy the properties required by the data flow analysis frameworks; that is, they should include the identity function and satisfy the properties  of closure under composition  and monotonicity. If they do then there exist interactive algorithms that    solve them.

In Appendix~\ref{ap:tfproperties}, we show that DSF transfer functions include an identity function and is closed under composition. Furthermore, we show that DSF transfer functions are distributive, which guarantees that they are monotone and finds the MOP solution for DSF transfer functions.

%%%%%%%%%%%%%%%%%%%%%%%%%%%%%%%%%%
%%%%%%%%%%%%%%%%%%%%%%%%%%%%%%%%%%

%%%-------------------------------%%%
\section{Solving the local DUA-node subsumption problem}
\label{sec:sa}
We present an interactive algorithm, described in Algorithm~\ref{alg}, to solve the DUA subsumption problem by means of DSF. It is adapted from classical algorithms \cite{lam07:ing} using the  above transfer functions. Although the DSF solution consists of sets \textbf{IN}($n$) and \textbf{OUT}($n$), we are interested in the final values of the \textbf{Covered}($n$) sets. They contain the DUAs that are covered at node $n$ whenever it is reached from any path beginning at the start node $s$  of $G$. In other words, it finds the \textit{local} DUA-node subsumption; that is,  the annotation presented in Figure~\ref{fig:DUAnodeMax}.
We call Algorithm~\ref{alg} \textit{Subsumption Algorithm} (SA).

\begin{algorithm}[ht]
  \KwIn{Flow graph $G(N,E, s, e)$ of program $P$; sets \textbf{Disabled}($n$), \textbf{Sleepy}($n$), \textbf{PotCovered}($n$), and \textbf{Born}($n$), all DUAs required to test $P$}
  \KwOut{\textbf{Covered}($n$) set for every node $n$}
  \vspace*{0.25cm}
   \textbf{IN}($s$) = $\varnothing$ where $s$ is the start node\;
   \textbf{OUT}($s$) =  \textbf{Born}($s$) where $s$ is the start node\;
   \textbf{Covered}($s$) = $\varnothing$ where $s$ is the start node\;
  \For{each node $n$ other than the start node $s$}{
  \textbf{OUT}($n$) = all DUAs of program $P$\;
  \textbf{Covered}($n$) = all DUAs of program $P$\;
  }
\While{changes to any \textbf{OUT} occur}{
  \For{each node $n \in N$}{
     \textbf{IN}($n$) = $\bigcap_{p \in PRED(n)}$ \textbf{OUT}($p$)\;
     \textbf{CurSleepy} = $\bigcup_{p \in PRED(n) \text{and} (p,n) \text{is not a back edge}}$ \textbf{Sleepy}($p$)\;
     \textbf{Covered}($n$) = $\bigcap_{p \in PRED(n)}$ \textbf{Covered}($p$) $\bigcup$ [(\textbf{IN}($n$) - \textbf{CurSleepy}) $\bigcap$ \textbf{PotCovered}($n$)]\;
   \textbf{OUT}($n$) = \textbf{Born}($n$) $\bigcup$ [\textbf{IN}($n$) - \textbf{Disabled}($n$)] $\bigcup$ \textbf{Covered}($n$)\;
    }
    }
      \For{each node $n \in N$}{
     \textbf{IN}($n$) = $\bigcap_{p \in PRED(n)}$ \textbf{OUT}($p$)\;
     \textbf{CurSleepy} = $\bigcup_{p \in PRED(n) \text{and} (p,n) \text{is not a back edge}}$ \textbf{Sleepy}($p$)\;
     \textbf{Covered}($n$) = $\bigcap_{p \in PRED(n)}$ \textbf{Covered}($p$) $\bigcup$ [(\textbf{IN}($n$) - \textbf{CurSleepy}) $\bigcap$ \textbf{PotCovered}($n$)]\;
    }
\Return{\textbf{Covered}($n$) for every node $n$}
\vspace*{0.25cm}
  \caption{Subsumption algorithm.}
\label{alg}
\end{algorithm}

Lines 1-12 comprise the iterative algorithm to solve DSF transfer functions \cite{lam07:ing}; that is, to find the \textbf{IN} and \textbf{OUT} sets. Lines 13-16 update the \textbf{Covered}($n$) sets.  \textbf{OUT}($n$) has already converged to its final value after leaving the while-loop at Line 7, but  \textbf{Covered}($n$) needs to be updated with the final values of \textbf{OUT}($p$).

When the transfer functions hold the properties required by a monotone data flow framework, interactive algorithms can find the MFP  solution to the particular framework.
By showing that the set of DFS transfer functions contains the identity function and are closed under composition and distributive, these algorithms will find the MOP solution for sets \textbf{IN}($n$) and \textbf{OUT}($n$) of  DSF transfer functions (see Appendix~\ref{ap:tfproperties}). 

However, it rests to show that \textbf{IN}($n$) and \textbf{OUT}($n$) sets contain, respectively, the set of covered DUAs or available to be covered  at the entrance and exit of a node $n$ when it is reached by any path starting at node $s$. We show  in Appendix~\ref{ap:saproof} that SA finds the local DUA-node subsumption.

%%%-------------------------------%%%
\subsection{Complexity}

The complexity of SA is given by the number of iterations needed to find the solution of DSF; that is, to terminate the execution of the while-loop of Algorithm~\ref{alg}.  DSF is a data flow analysis framework which has a set  of values $V$, given by the power set of  all DUAs ($U$) required to test a program $P$, and a meet operation $\wedge$, given by set intersection operation ($\cap$). $V$ and the meet operation
constitutes a semi-lattice.

In the worst case, the cost to find the solution for DSF is the product of the height of the semi-lattice and the number of nodes of the flow graph \cite{lam07:ing}. The height of DSF semi-lattice is the number of DUAs. However, DSF shares a characteristic with other practical data flow analysis problems like reaching definitions and available expressions. The value of $V$ (fact)  at each node --- in the DSF case, the covered or available DUAs --- propagates along cycle-free paths.

This property can be expressed as follows.  If a DUA $D$ is removed from \textbf{IN}($n$) or \textbf{OUT}($n$) then there is a cycle-fre path from node $s$ to the beginning or end of  node $n$, respectively, along which $D$ is never covered, or there is an cycle-free path from the node where $D$'s variable is redefined to node $n$.
 To be in \textbf{OUT}($n$), a DUA $D$ should be covered  or available  in all paths from $s$ to $n$, then it suffices an cycle-free path in which $D$ is not covered from $s$ or an cycle-free path from the redefinition of $D$'s variable to $n$ to remove it.

The visit of the nodes in iterative algorithms can be determined in such a way that the information is passed along cycle-free paths in few iterations. If the nodes are visited in a depth-first ordering, the retreating edges\footnote{Edge ($n'$, $n$) is a retreating edge if  node $n'$ has a higher depth-first ordering than $n$. Depth-first ordering is also known as reverse postorder (rPostorder) \cite{hec77:ing}. Every back edge is a retreating edge, but not every retreating edge is a back edge \cite{lam07:ing}}. will be visited at the end, after the information has gone through the cycle-free paths.

Using this approach, the number of iterations will be no greater than the number of retreating edges plus two. However, the number of retreating edges is never greater than the depth of loop nesting in a flow graph  \cite{lam07:ing}. In practice, the nesting of loops seems limited to a small constant for intra-procedural flow graphs \cite{knu71,ryd86}.

However, SA has the additional cost of finding the dominance relationship to determine the back edges. Luckily, the dominance relationship is also modeled as a data flow analysis problem with the same property of propagating its fact (the dominator nodes) along cycle-free paths. Thus, the dominance relationship is found at the same cost. 

As a result, the cost of SA tends to be linear to the number of nodes of the flow graph or $\approx O(\mid N \mid)$.

%%%%%%%%%%%%%%%%%%%%%%%%%%%%%%%%%%
%%%%%%%%%%%%%%%%%%%%%%%%%%%%%%%%%%
%\section{Related work}
%\label{sec:related}

%%%%%%%%%%%%%%%%%%%%%%%%%%%%%%%%%%
%%%%%%%%%%%%%%%%%%%%%%%%%%%%%%%%%%
\section{Conclusions}
\label{sec:concl}
We presented the formal description of the Data flow Subsumption Framework (DSF). DSF models the problem of finding the definition-use associations (DUAs) covered whenever a particular point $p$ (e.g, a node) of a program is reached as a data flow analysis problem. We call this problem \textit{local DUA-node subsumption}.

DSF is a distributive data flow analysis framework. This property allows iterative algorithms to find the Meet-Over-All-Paths (MOP) solution for DSF transfer functions, which implies that the results at a point $p$ are valid for all paths that reach $p$. The  Subsumption Algorithm (SA) presented in this work is an incarnation of one of them. SA utilizes DSF transfer functions to calculate efficiently the local DUA-node subsumption at $\approx O(\mid N \mid)$ cost. 

This result has practical implications since SA allows  to find the subsumption of DUAs by edges (local DUA-edge subsumption)  and by DUAs (DUA-DUA subsumption) at  $\approx O(\mid N \mid)$ and  $\approx O(\mid U \mid  \mid N \mid)$ costs, respectively \cite{cha20a}, where  $N$ is the set of nodes of the program's flow graph and $U$ is the set of DUAs required to test it according to all-uses criterion. 

Previous solutions are significantly more expensive. Santelices and Harrold \cite{san07:ing} propose an algorithm to find DUA-edge subsumption at  $O(\mid U \mid)$ cost. Marré and Bertolino's \cite{marr96:ing,mar97,mar03:ing} and Jiang et al.'s \cite{ji18} solutions  are quadratic to the number of DUAs for the DUA-DUA subsumption, which hamper their application  at industrial settings. Experimental results suggest that SA's application on industry-like applications works at scale and is quite promising \cite{cha20a}.

\section*{Acknowledgment}

Marcos Lordello Chaim was supported by  grant \#2019/ 21763-9, São Paulo Research Foundation (FAPESP). 

The opinions, hypotheses and conclusions or recommendations expressed in this material are the responsibility of the authors and do not necessarily reflect the vision of FAPESP.

%% Both columns the same height on the last page
\balance

\bibliographystyle{./bibliography/IEEEtran}
%% Fixed typos in graDUA:off99, misc:DO-178C, misc:MiniMutantsICST2014, misc:MutantSubsumptionMu2014, misc:miniMutantsMut2016, graDUA:araujo14
%% (Jeff, 6 July)
%\bibliography{./bibliography/quali,./bibliography/dout,./bibliography/graDUA,./bibliography/misc}

\appendix

%%%-------------------------------%%%
\subsection{Properties of transfer functions}
\label{ap:tfproperties}

The set of transfer functions $F$ of a monotone data flow analysis framework should satisfy particular conditions or properties, namely, $F$ should include an identity function and all $f_n \in F$ should  be monotone and closed under composition.

The importance of these properties is to allow the aggregation of data flow information at path level.  For  a  path  $s, \ldots, n_i, \ldots  n$, the first property, the identity function, implies that  $I(x)$ is associated with an empty path which keeps the fact unaltered. Being the transfer functions closed under composition implies that the transfer function of the path: $f_{s, \ldots, n_i, \ldots, n}(x) = f_{n} \circ \ldots \circ f_{n_i} \ldots \circ f_{s}$(x) is also a transfer function.

If the functions belonging to $F$ are monotone then,   for all $x$ and $y$ in $V$, $x \leq y$ implies $f(x) \leq f(y)$, where  $f: V \rightarrow V \in F$. This property allows iterative algorithms in general find the \textit{maximum fixed solution} (MPF) for data flow analysis problems. A stronger property of a data flow analysis framework is the \textit{distributive condition} given by:

\[
f(x \wedge y) = f(x) \wedge f(y)
\]

for all $x$ and $y$ in $V$ and $f$ in $F$.

A distributive framework is necessarily monotonic. The closure under composition and the distributive condition have an important implication on the iterative algorithm that solves DSF.  Because  DFS holds these properties, the iterative algorithm will find the \textit{meet-over-all-paths} (MOP) for DSF, which means that the solution is valid for any path taken from the start node $s$ to a node $n$.

In the case of DSF, the meet operator $\wedge$ is the set intersection $\cap$ operation. To simplify the manipulation of the DSF transfer functions, we are going to represent them as:

\begin{equation}\label{eq:f}
f(x) = B \cup (x - D) \cup C \cup (x - CS) \cap P
\end{equation}

where $x$ is the value of \textbf{IN}($n$), $f(x)$ is \textbf{OUT}($n$),  and $B$, $D$, $C$, $CS$, and $P$ are, respectively, the values of \textbf{Born}($n$), \textbf{Disabled}($n$),  the result of the intersection of the \textbf{Covered}($p$) where $p$ is a predecessor of $n$,  \textbf{CurSleepy}, and \textbf{PotCovered}($n$). All sets have fixed values, excepting the set $x$.

Before proving the  properties above, we highlight the fixed values of the DSF transfer functions. \textbf{Born}($n$), \textbf{Disabled}($n$), and \textbf{PotCovered}($n$) are fixed values associated to node $n$. \textbf{CurSleepy} is calculated by the union of the \textbf{Sleepy} sets of the predecessors of $n$. Since \textbf{Sleepy} sets are constant and associated with the predecessors of $n$, \textbf{CurSleepy} sets could be calculated beforehand for every node $n$. %  because they are constant values.

A less obvious fixed value is the result of $\bigcap_{p \in PRED(n)} \text{\textbf{Covered}}(p)$. Although the value of the \textbf{Covered}($n$) might change in every iteration of the algorithm, the value of the \textbf{Covered} sets of $n$'s predecessors are fixed when \textbf{OUT}($n$) is calculated as the output of a transfer function $f(x)$. In this sense,  \textbf{Covered}($p$) sets are as fixed as \textbf{Born}($n$), \textbf{Disable}($n$), \textbf{Sleepy}($n$), \textbf{PotCovered}($n$), and \textbf{OUT}($p$) during the calculation of \textbf{OUT}($n$) in a particular iteration.

In what follows, we show that  DSF transfer functions  comply with the properties that characterize it a distributive data flow analysis framework.

\subsubsection{Identity function}

For a identity function $I(x)$, such that $I(x) = x$,  be part of $F$, it  suffices $B$, $D$, $C$, $CS$, and $P$ all  be the empty set.

\subsubsection{Closure under composition}

 The set $F$ of transfer functions is closed under composition iff, for any two functions $f$ and $g$ in $F$, the function $h$ defined by $h(x) = g(f(x))$ is in $F$. To show the closure under composition, let us suppose we have two functions:
\begin{equation}\label{eq:f1}
f_1(x) = B_1 \cup (x - D_1) \cup C_1 \cup (x - CS_1) \cap P_1
\end{equation}
  and
\begin{equation}\label{eq:f2}
f_2(x) = B_2 \cup (x - D_2) \cup C_2 \cup (x - CS_2) \cap P_2.
\end{equation}
 Then the composition of $f_2 \circ f_1$ will be:
\begin{quote}
$f_2(f_1(x)) = B_2 \cup ((B_1 \cup (x - D_1) \cup (C_1 \cup (x - CS_1) \cap P_1) - D_2) \cup C_2 \cup ((B_1 \cup (x - D_1) \cup C_1 \cup (x - CS_1) \cap P_1) - CS_2) \cap P_2$.
\end{quote}

Note that $ A \cup B \cup C - D$ is equivalent to $ (A - D) \cup (B - D) \cup (C - D)$.
We can rewrite $f_2(f_1(x))$ in the following steps.

\begin{enumerate}
\item $f_2(f_1(x)) = B_2 \cup (B_1 - D_2 \cup (x - D_1)  - D_2 \cup (C_1 \cup ((x - CS_1) \cap P_1) - D_2)) \cup (C_2 \cup (B_1 - CS_2 \cup (x - D_1) - CS_2 \cup (C_1 \cup (x - CS_1) \cap P_1 - CS_2)) \cap P_2)$.

Note that $ A - B - C$ is equivalent to $ A - (B \cup C)$.

\item $f_2(f_1(x)) = B_2 \cup (B_1 - D_2) \cup (x - (D_1 \cup D_2)) \cup (C_1 -D_2 \cup (x - CS_1) \cap P_1 - D_2)) \cup (C_2 \cup ((B_1 - CS_2) \cup (x - (D_1 \cup CS_2)) \cup (C_1 \cup (x - CS_1) \cap P_1 - CS_2)) \cap P_2)$.

Note that $(x - CS_1) \cap P_1 - D_2$ is equivalent to  $(x - (CS_1 \cup D_2)) \cap P_1$ and $(x - CS_1) \cap P_1 - CS_2$ to $(x - (CS_1\cup CS_2)) \cap P_1$.

%\item $f_2(f_1(x)) = B_2 \cup (B_1 - D_2) \cup (x - (D_1 \cup D_2)) \cup ((C_1 - D_2) \cup (x - (CS_1 \cup D_2)) \cap P_1) \cup (C_2 \cup ((B_1 - CS_2) \cup (x - (D_1 \cup CS_2)) \cup ((C_1 - CS_2) \cup (x - (CS_1\cup CS_2)) \cap P_1) \cap P_2)$.

\item $f_2(f_1(x)) = B_2 \cup (B_1 - D_2) \cup (x - (D_1 \cup D_2)) \cup ((C_1 - D_2) \cup ((x - (CS_1 \cup D_2)) \cap P_1) \cup (C_2 \cup ((B_1 - CS_2) \cup (x - (D_1 \cup CS_2)) \cup ((C_1 - CS_2) \cup (x - (CS_1 \cup CS_2)) \cap P_1) \cap P_2)$.

\item $f_2(f_1(x)) = B_2 \cup (B_1 - D_2) \cup (x - (D_1 \cup D_2)) \cup ((C_1 - D_2) \cup ((x - (CS_1 \cup D_2)) \cap P_1) \cup (C_2  \cup ((C_1 \cup B_1 - CS_2) \cup (x - (D_1 \cup CS_2)) \cup (x - (CS_1 \cup CS_2)) \cap P_1) \cap P_2)$.

\item  $f_2(f_1(x)) = B_2 \cup (B_1 - D_2) \cup (x - (D_1 \cup D_2)) \cup ((C_1 - D_2) \cup ((x - (CS_1 \cup D_2)) \cap P_1) \cup (C_2  \cup ((C_1 \cup B_1 - CS_2) \cap P_2 \cup (x - (D_1 \cup CS_2))\cap P_2 \cup (x - (CS_1 \cup CS_2)) \cap P_1 \cap P_2)$.

%Note that $(C_1 \cup B_1 - CS_2) \cap P_2 \subseteq C_2$.

\item $f_2(f_1(x)) = B_2 \cup (B_1 - D_2) \cup (x - (D_1 \cup D_2)) \cup (C_1 - D_2) \cup C_2 \cup (C_1 \cup B_1 - CS_2) \cap P_2   \cup (x - (CS_1 \cup D_2)) \cap P_1 \cup (x - (D_1 \cup CS_2))\cap P_2 \cup (x - (CS_1 \cup CS_2)) \cap P_1 \cap P_2$.

Note that $(x - (CS_1 \cup D_2)) \cap P_1$ is equivalent to $(x - D_2) \cap P_1  \cup (x - CS_1) \cap P_1$.

\item $f_2(f_1(x)) = B_2 \cup (B_1 - D_2) \cup (x - (D_1 \cup D_2)) \cup (C_1 - D_2) \cup C_2  \cup (C_1 \cup B_1 - CS_2) \cap P_2 \cup (x - D_2) \cap P_1  \cup (x - CS_1) \cap P_1  \cup (x - (D_1 \cup CS_2))\cap P_2 \cup (x - (CS_1 \cup CS_2)) \cap P_1 \cap P_2$.

Note that $(x - (CS_1 \cup CS_2)) \cap P_1 \cap P_2 \subseteq  (x - CS_1) \cap P_1$.

\item  $f_2(f_1(x)) = B_2 \cup (B_1 - D_2) \cup (x - (D_1 \cup D_2)) \cup (C_1 - D_2) \cup C_2 \cup (C_1 \cup B_1 - CS_2) \cap P_2 \cup (x - D_2) \cap P_1  \cup (x - CS_1) \cap P_1  \cup (x - D_1)\cap P_2 \cup (x - CS_2) \cap P_2$.

\item $f_2(f_1(x)) = B_2 \cup (B_1 - D_2) \cup (x - (D_1 \cup D_2)) \cup (C_1 - D_2) \cup C_2 \cup (C_1 \cup B_1 - CS_2) \cap P_2   \cup (x - (CS_1 \cup D_2)) \cap P_1 \cup (x - (D_1 \cup CS_2))\cap P_2$.

Note that $(x - (CS_1 \cup D_2)) \cap P_1 \cup (x - (D_1 \cup CS_2))\cap P_2$ is equivalent to $(x - (CS_1 \cup D_2)) \cup (x - (D_1 \cup CS_2)) \cap (P_1 \cup P_2$).

\item $f_2(f_1(x)) = B_2 \cup (B_1 - D_2) \cup (x - (D_1 \cup D_2)) \cup (C_1 - D_2) \cup C_2 \cup (C_1 \cup B_1 - CS_2) \cap P_2 \cup (x - (CS_1 \cup D_2)) \cup (x - (D_1 \cup CS_2)) \cap (P_1 \cup P_2$).

\item $f_2(f_1(x)) = B_2 \cup (B_1 - D_2) \cup (x - (D_1 \cup D_2)) \cup (C_1 - D_2) \cup C_2 \cup (C_1 \cup B_1 - CS_2) \cap P_2 \cup (x - (D_1 \cup D_2 \cup CS_1 \cup  CS_2)) \cap (P_1 \cup P_2$).

\end{enumerate}

We can rename $f_2(f_1(x))$ to $f_{c}(x)$ and also do the following renaming:
\begin{itemize}
\item $B_c$ = $B_2 \cup (B_1 - D_2)$;
\item $D_c$ =$D_1 \cup D_2$;
\item $C_c$ =  $ (C_1 - D_2) \cup C_2 \cup (C_1 \cup B_1 - CS_2) \cap P_2$;
\item $CS_c$ = $ D_1 \cup D_2 \cup CS_1 \cup  CS_2$; and
\item $P_c$ = $ P_1 \cup P_2$.
\end{itemize}
so that $f_c(x)$ = $B_c \cup (x - D_c) \cup C_c \cup (x - CS_c) \cap P_c$, the result of the composition of $f_1(x)$ and $f_2(x)$, belongs to the set of transfer functions $F$.

Using that $fc(x) = f_1 \circ f_2 (x) \in F$, one can show  by induction that the composition $f_{s, \ldots, n_i, \ldots, n} = f_{n} \circ \ldots \circ f_{n_i} \ldots \circ f_{s}$ also  is part of $F$.

\subsubsection{Distributive condition}

Let $y$ and $z$ be sets of DUAs in DSF and $f(x) = B \cup (x - D) \cup C \cup (x - CS) \cap P$ a transfer function belonging  to $F$. We show the distributive condition of DSF by verifying that:

\begin{quote}
$f(y \cap z) = B \cup (y \cap z - D) \cup C \cup (y \cap z - CS) \cap P$; and that

$f(y) \cap f(z) = (B \cup (y  - D) \cup C \cup (y  - CS) \cap P) \cap (B \cup (z  - D) \cup C \cup (z  - CS) \cap P)$

are equal.

\end{quote}

We start by rewriting  $f(y) \cap f(z)$. Note that $B$ and $C$ occurs in both $f(y)$ and $f(z)$, so they can be factored out.

\begin{quote}

$f(y) \cap f(z) = B \cup  C \cup ((y  - D) \cup (y  - CS) \cap P) \cap ((z  - D)  \cup (z  - CS) \cap P)$.

\end{quote}
Note that $(A_1 \cup B_1) \cap (A_2 \cup B_2)$ is equivalent to $(A_1 \cap A_2) \cup (B_1 \cap B_2)$. Then
\begin{quote}

$f(y) \cap f(z) = B \cup  C \cup ((y  - D)  \cap (z  - D)) \cup ((y  - CS) \cap P)   \cap (z  - CS) \cap P))$.

\end{quote}

which leads to

\begin{quote}

$f(y) \cap f(z) = B \cup  C \cup (y \cap z - D)   \cup ((y \cap z  - CS) \cap P) =  B \cup  (y \cap z - D) \cup C \cup  (y \cap z  - CS) \cap P = f(y \cap z)$.

\end{quote}

Thus, the distributive condition holds for DSF.

%%%-------------------------------%%%
\subsection{Proof of the Subsumption Algorithm}
\label{ap:saproof}

\begin{thm}
\label{thmalg}
The Subsumption Algorithm (SA), Algorithm \ref{alg},  is correct and  finds the DUAs covered at node $n$ when it is reached by any path starting at the start node $s$. That is, the sets \textbf{Covered}($n$) contain the DUAs covered in all paths $s \ldots n$ where $s$ is the start node and $n$ is a node of the program's flow graph.
\end{thm}

\begin{proof} In the proof presented below, we utilize the notation \textbf{OUT}($n$)$^i$ (extensive to the other sets) to represent the value of set \textbf{OUT}($n$) after the \textit{i-th} iteration of the while-loop of Algorithm~\ref{alg}.  \textbf{OUT}($n$)$^0$ refers to the value of the set after the initialization at lines 2--6 before the first iteration of the while. Additionally, we refer to $U$ as the set of all DUAs required to test program $P$ according to the all-uses criterion.

\textbf{Termination.}
Firstly, we show that the algorithm terminates by proving that  sets \textbf{OUT}($n$) eventually become constant and the while-loop condition becomes false.
To achieve such a goal, we show by induction that \textbf{OUT}($n$)$^{i+1}$ $\subseteq $  \textbf{OUT}($n$)$^{i}$, and eventually they become constant. To facilitate the proof, we also show that  \textbf{Covered}($n$)$^{i+1}$ $\subseteq $  \textbf{Covered}($n$)$^{i}$.

\textit{Basis.} \textbf{OUT}($n$)$^{1}$ $\subseteq $  \textbf{OUT}($n$)$^{0}$ and \textbf{Covered}($n$)$^{1}$ $\subseteq $  \textbf{Covered}($n$)$^{0}$.

 %The value  of \textbf{OUT}($s$), where $s$ is the start node, remains constant and equals to \textbf{Born}($s$). That is,  \textbf{OUT}($s$)$^{i+1}$ $\subseteq $  \textbf{OUT}($s$)$^{i}$.

\textit{Basis proof.}

Let us consider first the case for the start node $s$. The values of \textbf{IN}($s$)$^{0}$ and \textbf{OUT}($s$)$^{0}$ are initiated with $\varnothing$ and  with \textbf{Born}($s$)  at lines 1 and 2, respectively.  \textbf{IN}($s$) remains unaltered when line 9 is executed because $s$ does not have predecessors; thus, \textbf{IN}($s$)$^{1} = \varnothing$.  Due to the same reason, the intersection of  \textbf{Covered} sets is also empty at line 11.
 Because   \textbf{IN}($s$)$^{1}$  is $\varnothing$,  (\textbf{IN}($n$)$^{1}$  - \textbf{CurSleepy}) $\bigcap$ \textbf{PotCovered}($n$) is also $\varnothing$; as a result, \textbf{Covered}($s$)$^{1}$ becomes empty.  Thus, \textbf{Covered}($s$)$^{1}$ = \textbf{Covered}($s$)$^{0}$
= $\varnothing$; therefore, \textbf{Covered}($s$)$^{1}$ $\subseteq$ \textbf{Covered}($s$)$^{0}$.

 \textbf{OUT}($s$)$^{1}$, in turn, is equal to \textbf{Born}($s$) at line 12 since (\textbf{IN}($n$)$^{1}$ - \textbf{Disabled}($n$)) and \textbf{Covered}($s$)$^{1}$  are both $\varnothing$. Hence,  \textbf{OUT}($s$)$^0$ = \textbf{OUT}($s$)$^{1}$ =  \textbf{Born}($s$). Indeed, \textbf{OUT}($s$) is constant and equals to \textbf{Born}($s$); hence, \textbf{OUT}($s$)$^{1}$ $\subseteq $  \textbf{OUT}($s$)$^{0}$.

Let $n$ be a node such that $n \neq s$. At line 9, \textbf{IN}($n$)$^{1}$ is calculated by \textit{anding} the OUT$^{0}$ sets of its predecessors. All \textbf{OUT}($n$)$^{0}$ sets, though, are $U$, excepting \textbf{OUT}($s$)$^{0}$  whose value is \textbf{Born}($s$), and remains so.  Since \textbf{Born}($s$) represents DUAs that born  (become available) at node $s$, they are limited to $U$; that is, \textbf{Born}($s$) $\subseteq$ $U$. As a result,  \textbf{IN}($n$)$^{1}$, such that $n \neq s$, will be a subset of $U$ after the \textit{intersection} at line 9.

The intersection of  \textbf{Covered}($p$)$^{0}$ at line 11 where $p \in $ PRED($n$), is a subset of $U$ for the same reason \textbf{IN}($n$)$^{1}$ is: \textbf{Covered}($s$)$^{0}$ is $\varnothing$ and remains so  and all \textbf{Covered}($n$)$^{0}$, $n \neq s$,  are initialized with $U$.  Thus, the \textit{and} of the Covered sets of the predecessors of $n$ is a subset of $U$. \textbf{CurSleepy} is a constant value comprising of only edge DUAs; thus, \textbf{CurSleepy} is a subset of $U$, and it only shrinks \textbf{IN}($n$). Hence, [(\textbf{IN}($n$) - \textbf{CurSleepy}) $\bigcap$ \textbf{PotCovered}($n$)] is also a subset of  $U$  because \textbf{IN}($n$)$^{1}$, \textbf{CurSleepy}, and \textbf{PotCovered}($n$) are all subsets of $U$. Therefore,  \textbf{Covered}($n$)$^{1}$ calculated at line 11 is a subset of $U$. Since \textbf{Covered}($n$)$^{0}$, $n \neq s$, was initialized with $U$ at line 6,  \textbf{Covered}($n$)$^{1}$ $\subseteq$ \textbf{Covered}($n$)$^{0}$.

At line 12, \textbf{OUT}($n$)$^{1}$ is calculated. We already know that \textbf{Covered}($n$)$^{1}$ and \textbf{Born}($n$) are limited to $U$. (\textbf{IN}($n$)$^{1}$ - \textbf{Disabled}($n$)) is also a subset of $U$ because \textbf{IN}($n$)$^{1}$ is, and \textbf{Disabled}($n$) ---  a constant comprising DUAs unavailable due to the redefinition of their variable  at $n$ --- diminishes \textbf{IN}($n$)$^{1}$. Hence, \textbf{OUT}($n$)$^{1}$ is a subset of $U$. However, \textbf{OUT}($n$)$^{0}$, $n \neq s$, was initialized with $U$ (line 5); therefore,  \textbf{OUT}($n$)$^{1} \subseteq $ \textbf{OUT}($n$)$^{0}$, which completes the proof of the Basis.

\textit{Induction.}  Assuming  \textbf{OUT}($n$)$^{i}$ $\subseteq $  \textbf{OUT}($n$)$^{i-1}$ and   \textbf{Covered}($n$)$^{i}$ $\subseteq $  \textbf{Cover\-ed}($n$)$^{i-1}$, then \textbf{OUT}($n$)$^{i+1}$ $\subseteq $  \textbf{OUT}($n$)$^{i}$ and \textbf{Covered}($n$)$^{i+1}$ $\subseteq $  \textbf{Covered}($n$)$^{i}$.

\textit{Induction proof.}

At line 9, \textbf{IN}($n$)$^{i+1}$ is calculated by \textit{anding} the \textbf{OUT}($p$)$^{i}$ sets of its predecessors $p$.  Since \textbf{OUT}($p$)$^{i}$ $\subseteq$ \textbf{OUT}($p$)$^{i-1}$ and \textbf{IN}($n$)$^{i}$ is the \textit{anding} of \textbf{OUT}($p$)$^{i-1}$ then \textbf{IN}($n$)$^{i+1}$ $\subseteq$ \textbf{IN}($n$)$^{i}$.

At line 11, \textbf{Co\-ver\-ed}($n$)$^{i+1}$ is calculated by the union of the intersection of the \textbf{Covered} sets of its predecessors with  [(\textbf{IN}($n$) - \textbf{CurSleepy}) $\bigcap$ \textbf{PotCovered}($n$)]. \textbf{Co\-ver\-ed}($n$)$^{i+1}$ does not grow due to the contribution of the covered sets of its predecessors since we assume \textbf{Co\-ver\-ed}($p$)$^{i}$  $\subseteq$ \textbf{Covered}($p$)$^{i-1}$, which implies  $\bigcap_{p \in PRED(n)}$ \textbf{Covered}($p$)$^{i}$ $\subseteq$ $\bigcap_{p \in PRED(n)}$ \textbf{Co\-ver\-ed}($p$)$^{i-1}$.

\textbf{Co\-ver\-ed}($n$)$^{i+1}$ also receives the result of  [(\textbf{IN}($n$)$^{i+1}$ - \textbf{CurSleepy}) $\bigcap$ \textbf{PotCo\-ver\-ed}($n$)]. Since \textbf{IN}($n$)$^{i+1}$ $\subseteq$ \textbf{IN}($n$)$^{i}$ and \textbf{CurSleepy} and \textbf{PotCovered}($n$)] are constant values, we can conclude that [(\textbf{IN}($n$)$^{i+1}$ - \textbf{CurSleepy}) $\bigcap$ \textbf{PotCover\-ed}($n$)] $\subseteq$ [(\textbf{IN}($n$)$^{i}$ - \textbf{CurSleepy}) $\bigcap$ \textbf{PotCover\-ed}($n$)]. Therefore,  \textbf{Co\-ver\-ed}($n$)$^{i+1}$  $\subseteq$ \textbf{Co\-ver\-ed}($n$)$^{i}$.

\textbf{OUT}($p$)$^{i+1}$ is calculated by \textbf{Born}($n$) $\bigcup$ [\textbf{IN}($n$)$^{i+1}$ - \textbf{Disabled}($n$)] $\bigcup$ \textbf{Cover\-ed}($n$)$^{i+1}$ at line 12. As \textbf{Born}($n$) and \textbf{Disabled}($n$)  are constant values and \textbf{IN}($n$)$^{i+1}$ $\subseteq$ \textbf{IN}($n$)$^{i}$ and \textbf{Co\-ver\-ed}($n$)$^{i+1}$  $\subseteq$ \textbf{Co\-ver\-ed}($n$)$^{i}$, \textbf{OUT}($p$)$^{i+1}$ will also be a subset of \textbf{OUT}($p$)$^{i}$; that is, \textbf{OUT}($p$)$^{i+1}$ $\subseteq$ \textbf{OUT}($p$)$^{i}$, which completes the proof of the Induction.

It remains to prove that sets \textbf{OUT}($n$) eventually stop changing. Let us suppose \textbf{OUT}($n$)$^{i} \neq \varnothing$ $\subseteq$ \textbf{OUT}($n$)$^{i-1}$. In the \textit{(i+1)-th} iteration of the while-loop, we have two options:  (1)  \textbf{OUT}($n$)$^{i+1}$ $\subset$ \textbf{OUT}($n$)$^{i}$; or (2) \textbf{OUT}($n$)$^{i+1}$ = \textbf{OUT}($n$)$^{i}$. In case (1),  \textbf{OUT}($n$) is reduced even further; in case (2), \textbf{OUT}($n$) does not change. However,   in the next iterations, if \textbf{OUT}($n$)  keeps shrinking, it will eventually be  $\varnothing$; and if \textbf{OUT}($n$) $\neq$ $\varnothing$ keeps the same value, it has achieved its final value. Both cases lead to the falsehood of the loop condition, terminating the execution of the algorithm.

Lines 13-16  always terminate when the last node of the flow graph is visited.

\textbf{Correctness.}

The  solution of DSF---the correct final values of \textbf{IN}($n$) and \textbf{OUT}($n$)---takes into account two cases described below. We adapted them from the guidelines to prove the available expressions framework \cite{lam07:ing}.

\begin{enumerate}
\item If a DUA $D$ is removed from \textbf{IN}($n$) or \textbf{OUT}($n$) then there is a path from node $s$ to the beginning or end of  node $n$, respectively, along which $D$ (a) is never covered, or (b) after being available for coverage, its variable might be redefined.

\item If a DUA $D$ remains in  \textbf{IN}($n$) and \textbf{OUT}($n$), then along every path from node $s$ to the beginning or end of node $n$, respectively, (a) $D$ is covered or (b) is available for coverage.
\end{enumerate}

We prove below that  SA (Algorithm~\ref{alg}) solves DSF and that the \textbf{Covered}($n$) sets contain all DUAs covered at a node $n$ by all paths that reaches $n$.

\textit{Correctness proof}.

For every node $n$ of a flow graph $G$ = ($N$, $E$, $s$, $e$) of a program $P$, there exists at least one path $\pi = (s,...,n)$ from $s$, the start node of $G$, to $n \in N$. Let $D$($d$,$u$,$X$) or $D$($d$, ($u'$,$u$), $X$)  be a DUA required for testing $P$ according to all-uses criterion.

Before we start analyzing the two cases, we discuss the values of the sets at the start node $s$. As shown in the termination's proof, \textbf{IN}($s$), \textbf{OUT}($s$), and \textbf{Covered}($s$) are initialized, respectively, with $\varnothing$, \textbf{Born}($s$), and $\varnothing$. These values remain unaltered until the end of the while-loop.  These values are correct: There is no covered or available DUA before the start node; thus, \textbf{IN}($s$) = $\varnothing$; and \textbf{OUT}($s$) should contain only the DUAs that become available (are born) at node $s$ since no DUA is covered at the start node (\textbf{Covered}($s$) = $\varnothing$).

\textit{Case 1}.

Initially, all OUT sets have value $U$, excepting \textbf{OUT}($s$). \textbf{IN}($n$) is calculated at line 9 by the intersection of  \textbf{OUT}($p$), being $p$ a predecessor of $n$. So, if a DUA $D$ is removed from one of the \textbf{OUT}($p$) sets, due to Cases 1(a) or 1(b) above, it will not be part of  \textbf{IN}($n$). Therefore, the focus of the proof below is on the  \textbf{OUT}($n$) sets.

\textit{Case 1(a)}. This case deals with those DUAs that might not be included in \textbf{OUT}($n$) because they do not belong to \textbf{Covered}($n$).

All \textbf{Covered}($n$) are initialized with value $U$, excepting Cover\-ed($s$) = $\varnothing$ (lines 3 and 6). However,  a new \textbf{Covered}($n$)  is calculated (line 11) at every iteration of the algorithm by the intersection of the \textbf{Covered} sets of node $n$'s predecessors  \textit{plus} new DUAs covered at $n$ given by  formula [(\textbf{IN}($n$) - \textbf{CurSleepy}) $\bigcap$ \textbf{PotCovered}($n$)]. If a DUA $D \in$ \textbf{Covered}($n$), it will belong to \textbf{OUT}($n$) (line 12).

One possibility of not including $D$ in \textbf{Covered}($n$) is if $D$ $\not \in$ \textbf{Covered}($p$)  to  at least one of the \textbf{Covered}($p$)  sets. If so, the intersection of \textbf{Covered}($p$) sets  will remove $D$. As result, $D$ might be removed of \textbf{OUT}($n$) because it  is not covered in  a path $\pi = (s,...,p,n)$.  Another possibilities for not including $D$ is when $D \not \in$ \textbf{PotCovered}($n$) or $D \in$ \textbf{CurSleepy} at $n$.

These possibilities deals with the coverage of $D$ at node $n$; they consider that $D \in$ \textbf{IN}($n$). If  $D \in$ \textbf{CurSleepy}, then it is an edge DUA   that is not covered at $n$  because its edge ($u'$,$u$) is such that  $u = n$, but there are more than one path from $u'$ to $n$. This situation blocks the coverage of $D$ at $n$ since is not known from which path $n$ might have been reached.  Thus,  if $D \in$ [\textbf{IN}($n$) $\cap$ \textbf{CurSleepy}], it will not be included in \textbf{Covered}($n$) at line 12.

On the other hand, if  $D \not \in$ \textbf{PotCovered}($n$),  it is not covered at node $n$ because  $n$ is not the  use node of $D$. As result,    $D$ is not included in \textbf{Covered}($n$). If $D \not \in$ \textbf{Covered}($n$), it might not be included in \textbf{OUT}($n$) at line 12. This is so because $D$  is not covered in a path $\pi = (s,...,n)$

Therefore, a DUA $D$ might be removed of \textbf{OUT}($n$)  when there exists  a  path  $\pi = (s,...,n)$ in which $D$ is not covered, which proves  Case 1(a).

\textit{Case 1(b)}. This case describes situations in which a DUA is removed of \textbf{OUT}($n$) due to the redefinition of its variable.

Let us suppose that there exists a path $\pi = (s,...,k,...,p,n)$ such that a DUA $D \in$ \textbf{OUT}($k$) and $D \not \in$ \textbf{Covered}($k$). In other words, $D \in$ \textbf{OUT}($k$)  not because it has been previously covered in all paths $\pi'$= $(s,...,k)$, but because it is available after $k$.

$D$ might be removed from \textbf{OUT}($n$) either in the calculation of \textbf{IN}($n$) at line 9 or in the calculation of \textbf{OUT}($n$) at line 12. At line 9, $D$ will be removed because it does not belong to at least one of the \textbf{OUT} sets of $n$'s predecessors and will not make it to \textbf{IN}($n$). That means it has became unavailable in one of the its predecessors. Therefore, there is a redefinition of $D$'s variable  $X$ in one of the paths ($k,...,p$) where $p$ is a predecessor of $n$.

At line 12, $D$ will be removed of \textbf{OUT}($n$) if $D \in$ \textbf{Disabled}($n$); that is, variable $X$ is redefined at node $n$. Hence, the operation [\textbf{IN}($n$) - \textbf{Disabled}($n$)] will remove $D$ from  \textbf{OUT}($n$), which means that there is a redefinition of variable $X$ in the path ($k,...,n$).

Thus, a DUA $D \in$ \textbf{OUT}($k$) might be removed of \textbf{OUT}($n$) if there exists a redefinition of its variable $X$ in a path ($k,...,n$). This proves Case 1(b).

\textit{Case 2}.

The \textbf{IN}($n$) sets are calculated at line 9 by the intersection of  \textbf{OUT}($p$), being $p$ a predecessor of $n$. A DUA $D$ will be only part of  \textbf{IN}($n$), if it belongs to all \textbf{OUT}($p$) sets according to  Cases 2(a) or 2(b) above.  As a result, we focus the proof below on the \textbf{OUT}($n$) sets.

\textit{Case 2(a)}. This case deals with the condition a covered DUA remains in the \textbf{OUT} sets.

Let us suppose that a DUA $D \in$ \textbf{OUT}($n$) after executing line 12. One condition for that to happen is $D \in$ \textbf{Covered}($n$). $D$ will be part of \textbf{Covered}($n$) if it either is part of all \textbf{Covered} sets of $n$'s predecessors or $D$ is covered at $n$. The first possibility  implies that  $D$ is covered in all paths ($s,..., p,n$). The second possibility implies that $D \in$ \textbf{IN}($n$); that is, it is available in all paths ($s,..., p$), and also $D \not \in$ \textbf{CurSleepy}, $D \not \in$ \textbf{Disabled}($n$), and $D \in$ \textbf{PotCovered}($n$). These are the requirements for covering $D$ at node $n$. Hence, $D$ is covered in all paths ($s,...,,n$). This proves Case 2(a).

\textit{Case 2(b)}. This case describes  when a not covered DUA remains in the OUT sets.

Let us suppose that a $D \in$ \textbf{OUT}($n$) after executing line 12 but  $D \not \in$ \textbf{Covered}($n$). At line 12,  $D$ will not show up in \textbf{OUT}($n$) due to \textbf{Covered}($n$).  Any DUA $D \in$ \textbf{Born}($n$) is always included in \textbf{OUT}($n$) at line 12 because it becomes available exactly at $n$; hence, $D$ is available in all paths ($s,...,n$). Other possibility of inclusion in \textbf{OUT}($n$) is if $D \in$ \textbf{IN}($n$) and $D \not \in$ \textbf{Disabled}($n$). If so, $D$ is available in all paths  ($s,...,p$), because it belongs to \textbf{IN}($n$), and in  all paths  ($s,...,n$), because $D$ is not disabled at $n$. As a result, $D$ is available in all paths ($s,...,n$), which proves Case 2(b).

However, the result of Algorithm~\ref{alg} are the sets \textbf{Covered}($n$), not the sets \textbf{IN}($n$) and \textbf{OUT}($n$). Cases 1 and 2 shows that sets \textbf{IN}($n$) and \textbf{OUT}($n$) achieved the Meet-Over-All-Path (MOP) solution; that is, for all paths from the start node $s$ to $n$. However, having \textbf{OUT}($n$) achieved its final values does not guarantee that \textbf{Covered}($n$) has too because it is one step behind the \textbf{OUT}($n$). 

Lines 13-16 update \textbf{Covered}($n$) sets with the final values of the OUT($p$) where $p$ is a predecessor of $n$. After this step, a DUA $D \in$ \textbf{Covered}($n$) if it is covered in all paths ($s,...,n$); that is, all paths that reaches $n$. Therefore,  the Subsumption Algorithm  finds the local DUA-node subsumption.

\textbf{Completeness}.

There are two possibilities for a DUA $D$: It belongs to \textbf{OUT}($n$) or it does not. These two possibilities are dealt with  by Algorithm~\ref{alg} as discussed in Cases 1 and 2. Thus, it is complete.

\end{proof}

\end{document}